\documentclass[conference]{IEEEtran}

\usepackage{etex} 
\reserveinserts{28}

\usepackage{latexsym}
\usepackage{paralist}
\usepackage{microtype}

\usepackage{suffix}

\usepackage{xpunctuate}
\usepackage[all]{foreign} 
\usepackage{xfrac} 
\usepackage[style=american,english=american]{csquotes}

\usepackage{amsmath, amsthm, amssymb, amsfonts} 
\usepackage{thmtools, thm-restate}
\usepackage{mathtools}

\usepackage[full,small]{complexity}

\usepackage{xargs}

\PassOptionsToPackage{dvipsnames,usenames,table}{xcolor}

\usepackage{booktabs} 
\usepackage{float}
\usepackage{wrapfig}
\usepackage{floatflt}
\usepackage{framed}
\usepackage{ctable}
\usepackage{dcolumn}
\usepackage{adjustbox}

\usepackage{enumerate}
\usepackage{enumitem}
\usepackage{xspace}
\usepackage{xifthen}
\usepackage{fixltx2e}
\usepackage{url}

\usepackage[longend,vlined]{algorithm2e}

\makeatletter
\newcommand{\raisemath}[1]{\mathpalette{\raisem@th{#1}}}
\newcommand{\raisem@th}[3]{\raisebox{#1}{$#2#3$}}
\makeatother

\def\N{\mathbb{N}} 
\def\R{\mathbb{R}}	

\def\Q{\mathbb{Q}}	

\def\dir#1{\vec{#1}}

\renewcommand{\leq}{\leqslant}

\renewcommand{\geq}{\geqslant}

\renewcommand{\epsilon}{\varepsilon}

\renewcommand{\P}{\operatorname{\mathbb P}}
\renewcommand{\E}{\operatorname{\mathbb E}}

\usepackage{stmaryrd}

\theoremstyle{plain}
\newtheorem{lemma}{Lemma}

\newtheorem{theorem}{Theorem}
\newtheorem{corollary}{Corollary}

\theoremstyle{definition}
\newtheorem{definition}{Definition}

\newcommand*\varrule[1][0.4pt]{\leavevmode\leaders\hrule height#1\hfill\kern0pt}

\setlist[1]{labelindent=\parindent,leftmargin=*} 
\setlist{itemsep=0pt}
\setitemize[1]{label={\small\textbullet}}

\usepackage{hyphenat}
\def\Erdos{Erd\H{o}s\xspace}
\def\Renyi{R\'{e}nyi\xspace}

\def\Renyi{R\'{e}nyi\xspace}

\usepackage{subfig}

\newcommand{\SKG}{\ensuremath{\mathrm{SKG}}\xspace}
\newcommand{\SKGbf}{\ensuremath{\mathrm{\textbf{SKG}}}\xspace}
\newcommand{\rmatkerned}{\mathrm{R\text{-}M\kern -.1em AT}}
\newcommand{\rmatkernedbf}{\mathrm{\textbf{R\text{-}M\kern -.1em AT}}}
\newcommand{\RMAT}{\ensuremath{\rmatkerned}\xspace}
\newcommand{\RMATbf}{\ensuremath{\rmatkernedbf}\xspace}
\newcommand{\RMATErasure}{\ensuremath{\rmatkerned^\ominus}\xspace}
\newcommand{\RMATRethrow}{\ensuremath{\rmatkerned^\oplus}\xspace}
\newcommand{\RMATFlip}{\ensuremath{\rmatkerned^{\$}}\xspace}

\def\twobytwo(#1,#2,#3,#4){\ensuremath{\bigl( \begin{smallmatrix}
    #1  &  #2 \\%
    #3  &  #4%
  \end{smallmatrix} \bigr)}%
}

\newtheorem*{ChernoffHoeffding}{Chernoff--Hoeffding}

\begin{document}

\title{Asymptotic Analysis of Equivalences and Core-Structures in Kronecker-Style Graph Models}

\author{Anonymous}

\author{
\IEEEauthorblockN{
Alex J.~Chin \IEEEauthorrefmark{1},
Timothy D.~Goodrich \IEEEauthorrefmark{2},
Michael P.~O'Brien \IEEEauthorrefmark{2},\\
Felix Reidl \IEEEauthorrefmark{2},
Blair D.~Sullivan \IEEEauthorrefmark{2},
and Andrew van~der~Poel \IEEEauthorrefmark{2}}

\IEEEauthorblockA{\IEEEauthorrefmark{1}Department of Statistics,
Stanford University, Stanford, California 94305}

\IEEEauthorblockA{\IEEEauthorrefmark{2}Department of Computer Science, 
North Carolina State University, Raleigh, North Carolina 27606\\
\{tdgoodri, mpobrie3, fjreidl, blair\_sullivan, ajvande4\}@ncsu.edu,~ajchin@stanford.edu}
}

\setlength{\abovedisplayskip}{6pt}
\setlength{\belowdisplayskip}{5pt}
\setlength{\abovedisplayshortskip}{2pt}
\setlength{\belowdisplayshortskip}{2pt} 
\setlength{\floatsep}{5pt plus 2pt minus 2pt}
\setlength{\textfloatsep}{5pt plus 2pt minus 2pt}
\setlength{\intextsep}{5pt plus 2pt minus 2pt}

\newcommand{\tight}{%
    \setlength{\abovedisplayskip}{0pt}%
    \setlength{\belowdisplayskip}{2pt}%
    \setlength{\abovedisplayshortskip}{0pt}%
    \setlength{\belowdisplayshortskip}{2pt}%
}
\renewcommand{\baselinestretch}{.95} 

\maketitle

\begin{abstract}
Growing interest in modeling large, complex networks has spurred significant
research into generative graph models. These models not only allow users to
gain insight into the processes underlying networks, but also provide
synthetic data which allows algorithm scalability testing and addresses
privacy concerns. Kronecker-style models (\eg~SKG and R-MAT) are often used
due to their scalability and
ability to mimic key properties of real-world networks (\eg~diameter and
degree distribution). Although a few papers theoretically establish these
models' behavior for specific parameters, many claims used to justify the use
of these models in various applications are supported only by empirical
evaluations.  In this work, we prove several results using asymptotic analysis
which illustrate that empirical studies may not fully capture the true
behavior of the models.\looseness-1

Paramount to the widespread adoption of Kronecker-style models was the
introduction of a linear-time edge-sampling variant (R-MAT), which existing
literature typically treats as interchangeable with SKG. We prove that
although several R-MAT formulations are asymptotically equivalent, their
behavior diverges from that of SKG. Further, we show these results are
experimentally observable even at relatively small graph sizes. Second, we
consider a case where asymptotic analysis reveals unexpected behavior within a
given model. One of the criticisms of using Kronecker-style models has been
that they are unable to generate the deep core-structures commonly observed in
real-world data. We prove that in fact, for some parameter values, all the
Kronecker-style models generate graphs whose maximum core depth grows as a
function of the size of the network---including in the region of the
parameter space most commonly used in prior work. Our results also illustrate
why this behavior may be difficult to observe for moderate graph sizes, and
highlight the dangers of extrapolating model-wide claims from empirical
results.\looseness-1

\end{abstract}

%
\section{Introduction}
\noindent
The rapidly increasing availability of large relational data sets has brought
network science to the forefront of a diverse set of fields like business,
social sciences, natural sciences, and engineering. Due to privacy
restrictions and the desire to have testing data at larger scales, generating
synthetic data from random graph models to evaluate new algorithms or
techniques has become a common practice. A significant amount of research has
focused on creating generative models that produce networks whose properties
mimic those of real data sets. One popular family of such models, which we refer to 
as \emph{Kronecker-style models}, is based on using a small ``seed'' or initiator matrix to generate 
a fractal structure of edge probabilities. This family includes stochastic 
Kronecker graphs (\SKG)~\cite{leskovec2005realistic, leskovec2010kronecker}, the Recursive-MATrix (\RMAT) generator
~\cite{chakrabarti2004rmat}, and several variants of each~\cite{moreno2010tied, seshadhri2011depth}.

The widespread adoption~\cite{todorovic2012human, schmidt2009scalable, 
hill2009social, bader2008graph, sasaki2008security} of Kronecker-style models has been 
motivated by empirical evidence showing the 
generated networks replicate several important properties of real world networks, 
including degree and eigenvalue distributions, diameter, and 
density~\cite{leskovec2005realistic, leskovec2010kronecker}. Further, the initiator 
matrix parameters can be learned from a real-world network using the algorithm 
\textsc{KronFit}~\cite{leskovec2007scalable}; empirical evaluation shows that using the fitted parameters 
on a dozen datasets, synthetic \SKG graphs mimic real-world 
degree distributions and small diameters. To complement the work measuring 
properties of generated data, a number of papers have proven explicit expressions 
for computing the expected value of some graph 
invariant (e.g. degree distribution~\cite{groer2011mathematical} and number of 
isolated vertices~\cite{seshadhri2011depth}). Finally, a few papers have considered 
the limiting behavior of these models---characterizing the emergence of a giant
component and proving constant diameter~\cite{mahdian2007stochastic}, and proving that 
\SKG generally cannot generate graphs with a power-law degree distribution~\cite{kang2014properties}. 

Here we show that \emph{asymptotic analysis} of Kronecker-style models 
not only offers formal guarantees on limiting behavior, but 
practically-relevant restrictions on their usage. Specifically, we focus on two properties of these models: 
(1) equivalence/inequivalence among variants and (2) their core-periphery structure, as measured by degeneracy. 

Our first result addresses the common practice of using distinct variations of Kronecker-style
models interchangeably, despite the lack of formal proofs of equivalence in the literature. 
A recent paper of Moreno et al.~\cite{moreno2014scalable}
challenged these assumptions and proved that without careful consideration,
two Kronecker-style models will \emph{not} necessarily sample from the same
statistical distribution given analogous input parameters. We show in 
Section~\ref{sec:model-equiv} that in the limit, several widely-used 
variants of the \RMAT models are indeed equivalent 
(Theorem~\ref{theorem:erasure-rethrow-equiv}). However, we also prove
that their edge probabilities diverge from those of \SKG, and show these 
differences are experimentally observable even at relatively small graph sizes.

Our second contribution provides insight into the core-periphery structure of 
graphs generated by these models, as measured by the {degeneracy}.  
Low degeneracy means the graph has no region (subgraph) that is ``too dense'', 
and is an observed property of many real-world networks. Empirical
studies~\cite{seshadhri2011depth} have suggested that the degeneracy of
Kronecker-style models cannot grow large without increasing the number of
isolated vertices. Our proofs in Section~\ref{sec:degeneracy} disprove this, showing 
that for a fixed average degree, these models produce graphs whose degeneracy
grows asymptotically with the number of vertices \emph{irrespective of the
number of isolated vertices}. However, our results also imply that 
this asymptotic behavior is slow to appear, preventing the occurrence of
deep cores even for graphs with hundreds of thousands of vertices (and causing 
misleading empirical evidence).

%

\section{Preliminaries}\label{sec:prelims}
\noindent
We assume that all graphs are simple (no parallel edges or self-loops) and
undirected unless otherwise specified. Directed graphs will be denoted by an
arrow (e.g. $\dir G$). Let~$\mathcal G_n$ denote the set of all~$n$-vertex
graphs. A \emph{random graph model} is a sequence of probability
measures~$(\P_n)_{n \in \N}$ over the space~$(\mathcal G_n, 2^{\mathcal
G_n})$. For simplicity, we use~$\P$ as the probability measure with the
understanding that it refers to a concrete random graph model that will be
apparent from the context.

For convenience we consider $n$-vertex graphs whose vertices are numbered~$0$
to~$n-1$ and represented by binary bitstrings. This convention will allow us
to derive the (relative) probability of an edge in a Kronecker-style model
from the positions of ones in its endpoints. For bitstrings~$i,j$ of equal
length we use~$\#ab(ij)$ to denote the number of positions in which $a$ occurs
in $i$ when $b$ occurs in $j$, \ie 
\begin{align*}
\#ab(ij) = |\{x ~:~ i[x] = a ~\wedge~ j[x] = b\}| \text{ for } a,b \in \{0, 1\}.
\end{align*}

\subsection{Kronecker-style Models}

\noindent
We now define the Kronecker-style models, including a new formulation
\RMATFlip used in our analysis. For reference, Table \ref{table:model-def}
summarizes the notation defined below.\\

\begin{table}[!h]
\centering\small
\begin{tabular}{lp{5cm}p{4cm}}
  \toprule
  Symbol & Model Name\\
  \midrule
  \SKG & Stochastic Kronecker\\
  \RMATErasure & \RMAT with arc erasures \\
  \RMATRethrow & \RMAT with arc rethrows\\
  \RMATFlip & \RMAT simulated by coin-flips \\
  \bottomrule
\end{tabular}
\caption{Summary of Kronecker-style models.}\label{table:model-def}
\end{table}

\noindent \textbf{Stochastic Kronecker:} In 2005, Leskovec et al.
\cite{leskovec2005realistic} introduced the Stochastic Kronecker random graph
generator (\SKG) as a means of modeling real-world data. Taking an
\emph{initiator matrix} $M_1$ with values in the interval $[0, 1]$ (not
necessarily summing to $1$) and a natural number $k$, \SKG starts by
deterministically generating a probability matrix $M_k$, such that $M_k =
M_{k-1} \otimes M_1$, where $\otimes$ is the tensor (Kronecker) product. A
directed graph can then be sampled from $M_k$ by flipping one biased coin per
matrix entry to obtain an adjacency matrix. In keeping with prior work, we assume 
$M_1 = \twobytwo(a,b,c,d)$; such  a $2 \times 2$ initiator matrix has
been most widely adopted in the literature (including~\cite{seshadhri2011depth, groer2011mathematical, mahdian2007stochastic, kang2014properties}) after
experiments found it generates synthetic graphs that most closely match 
real-world data~\cite{leskovec2007scalable}.\\

\noindent \textbf{\RMATbf erasure and rethrow models:} Independent of \SKG, in
2004, Chakrabarti et al. \cite{chakrabarti2004rmat} introduced the Recursive
MATrix (\RMAT) graph generator. Similar to \SKG, \RMAT starts with an
initiator matrix $M$ and natural number $k$, with the restriction that $M =
\twobytwo(\alpha,\beta,\gamma,\delta)$ and $\alpha + \beta + \gamma + \delta =
1$; we will also assume (without loss of generality) that $\alpha \geq \beta,
\gamma, \delta$. A $2^k \times 2^k$ directed adjacency matrix is then
constructed by iteratively ``throwing'' $m$ arcs recursively into quadrants of
the adjacency matrix based on the probabilities from $M$; we call this method
the \emph{general R-MAT process}.

Formally, starting with a graph~$\dir G^0$ on~$2^k$ vertices with no arcs, at
each step~$1 \leq i \leq m$ we generate a random arc~$e_i$ by flipping two
biased coins~$C_1, C_2$ for $k$ rounds where $\P[C_1 = 1] = \gamma+\delta$
and~$\P[C_2 = 1] = \beta+\delta$. The head of~$e_i$ is the bitstring formed 
by concatenating the results of~$C_2$, and the tail is obtained using $C_1$. 
We then either add $e_i$ to the graph or rethrow the edge (detailed
below) to obtain $\dir G^i$.

The probability of an arc being selected in a single step is a function of the bitstrings of
its endpoints. For vertices~$u,v$ we define the \emph{weight} of the (potential)
arc~$uv$ as
\begin{align}
  \omega_{uv} = \alpha^{\#00(uv)} \beta^{\#01(uv)} \gamma^{\#10(uv)} \delta^{\#11(uv)}.
  \label{eqn:omega}
\end{align}
Note that the probability of an arc existing in $\dir G^m$ is not necessarily
its weight; rather, the probability can be computed from its weight and proper
model-dependent scaling.

Given that we are generating simple graphs and a thrown arc may land in an
occupied cell, we now define two existing implementations of the general \RMAT
process. In the erasure model (denoted by \RMATErasure), the repeated arc is ignored, 
resulting in a generated graph with strictly less arcs than the number thrown. 
In the rethrow model (denoted by \RMATRethrow), this arc is ``rethrown'' 
(flipping another $k$ pairs of coins) until it lands in an unoccupied cell.
These two models are not strictly identical, since the probability
distribution across unoccupied cells changes with each added arc in the
rethrow model.
\RMATErasure is equivalent to the original formulation of the \RMAT
model~\cite{chakrabarti2004rmat}, and \RMATRethrow is consistent with the
description of the \RMAT model in~\cite{leskovec2010kronecker}.\\

\noindent \textbf{Converting parameters between \SKGbf and \RMATbf:}\\
Historically, \RMAT has been treated as an $O(m)$ run time drop-in
replacement for \SKG (\eg in~\cite{leskovec2010kronecker, seshadhri2011depth,
pinar2011similarity}), but the details of converting parameters between models
require some care. Let $\twobytwo(a,b,c,d)$ be a \SKG initiator matrix, then
each arc~$uv$ is added independently at random with probability
\[
   a^{\#00(uv)} b^{\#01(uv)} c^{\#10(uv)} d^{\#11(uv)}
\]
and the expected number of arcs in the final graph is
\[
  m = (a+b+c+d)^k.
\]
This formulation suggests the following translation between the \RMAT and \SKG
parameters. Suppose we want a graph with $n = 2^k$ vertices and $m$ arcs. Let
$\mu = m / 2^k$ (i.e. the edge density), then we introduce a scaling
parameter~$\theta$ where $\twobytwo(a,b,c,d) =
\theta\twobytwo(\alpha,\beta,\gamma,\delta)$. To compute~$\theta$ we match up
the expected number of edges of the models:
\begin{align*}
  \mu \cdot 2^k = (a+b+c+d)^k = \theta^k(\alpha+\beta+\gamma+\delta)^k = \theta^k.
\end{align*}
Therefore our scaling parameter is~$\theta = 2\mu^{1/k}$.
Table~\ref{tab:rmat-converted} contains several \RMAT 
initiator matrices that were derived from \SKG initiators fitted to real-world
networks~\cite{leskovec2010kronecker} using this conversion.

Note that in order to satisfy
\SKG's constraint that $a \leq 1$ (and given that $\alpha > \beta, \gamma,
\delta$) we require that
\[
  \theta \alpha = \mu^{1/k} 2\alpha \leq 1 \implies \alpha \leq \frac{1}{2} \mu^{-1/k}.
\]
This latter term converges to~$1/2$ when~$k \to \infty$ and~$\mu$ is a 
constant independent of~$n$. Since it is generally accepted that real-world
networks are sparse, we will restrict ourselves to constant~$\mu$
in the rest of this paper. In conclusion,
the translation from \RMAT to \SKG parameters is possible
whenever~$\alpha < 1/2$, $\mu$ is a constant, and $k$ is large enough. 
This restriction on~$\alpha$ has another interpretation: for~$\alpha \geq 1/2$, 
\SKG~generates in expectation a sublinear number of arcs. Accordingly, in this regime \SKG~cannot possibly
match up with the \RMAT model.  \\

\noindent \textbf{A new \RMATbf model:} In addition to the parameter space
limitations, the arc generation process differs significantly between \SKG and
\RMAT. Specifically, arcs occur in \SKG independently while the existing $i-1$
arcs influence the placement of the $i^\text{th}$ arc in \RMATRethrow and
\RMATErasure. To study whether this difference in mechanics results in
dissimilar models, we introduce a new \emph{coin-flipping} model, \RMATFlip.
This model is not intended for practical usage since sampling from it requires
$O(n^2)$ iterations, but it is useful for mathematical analysis and to bridge
the gap between the previous \RMAT models and \SKG.

First, note that the probability of an arc~$uv$ occurring $t$ times in the \RMAT process follows the binomial law
\begin{align*}
  \P[|\{e_i \colon e_i = uv\}| = t]
    &= {m \choose t} (\omega_{uv})^t (1-\omega_{uv})^{m-t}
\end{align*}
where $\omega_{uv}$ is defined in Equation \ref{eqn:omega}. Therefore the arc~$uv$ exists after~$m$ arcs have been thrown with probability
\begin{align*}
  \P[e \in \dir G^{m}] = 1 - (1-\omega_{uv})^{m}.
\end{align*}

\noindent Utilizing this fact, we define the \RMATFlip model:
\begin{definition}[\RMATFlip]
  Given an initiator matrix $M = \twobytwo(\alpha, \beta, \gamma, \delta)$ with
  $\alpha, \beta, \gamma, \delta \geq 0$ and $\alpha + \beta + \gamma + \delta =
  1$, a natural number $k$, and a positive real number $\mu$, \RMATFlip
  generates a graph with $2^k$ vertices by flipping every potential arc~$uv$
  independently at random with probability $1 - (1-\omega_{uv})^m$.
\end{definition}

\def\netw#1{\textsc{#1}}
\begin{table}
  \begin{tabular}{llll}
    \toprule
    Network & $n$  & $\mu$ & \RMAT initiator \\ \midrule
\netw{AS-NEWMAN} &  $22963$ & $4.22$ & $[.432, .269; .269, .009]$ \\
\netw{AS-ROUTEVIEWS} &  $6474$ & $4.09$ & $[.442, .255; .255, .022]$ \\
\netw{BIO-PROTEINS} &  $4626$ & $6.40$ & $[.364, .275; .275, .031]$ \\
\netw{CA-DBLP} &  $425957$ & $6.33$ & $[.453, .139; .139, .260]$ \\
\netw{CA-GR-QC} &  $5242$ & $5.53$ & $[.435, .107; .107, .301]$ \\
\netw{CA-HEP-PH} &  $12008$ & $19.74$ & $[.401, .175; .175, .194]$ \\
\netw{CA-HEP-TH} &  $9877$ & $5.26$ & $[.441, .120; .120, .259]$ \\
\netw{EMAIL-INSIDE} &  $986$ & $32.58$ & $[.352, .272; .272, .091]$ \\
\midrule
\netw{ANSWERS} &  $598314$ & $3.07$ & $[.469, .181; .195, .117]$ \\
\netw{ATP-GR-QC} &  $19177$ & $1.36$ & $[.441, .124; .108, .285]$ \\
\netw{BLOG-NAT05-6M} &  $31600$ & $8.59$ & $[.433, .246; .217, .096]$ \\
\netw{BLOG-NAT06ALL} &  $32443$ & $9.83$ & $[.429, .248; .222, .095]$ \\
\netw{CIT-HEP-PH} &  $30567$ & $11.41$ & $[.422, .186; .151, .223]$ \\
\netw{CIT-HEP-TH} &  $27770$ & $12.70$ & $[.417, .185; .146, .226]$ \\
\netw{DELICIOUS} &  $205282$ & $2.13$ & $[.479, .157; .167, .187]$ \\
\netw{EPINIONS} &  $75879$ & $6.71$ & $[.444, .237; .213, .057]$ \\
\netw{FLICKR} &  $584207$ & $6.09$ & $[.455, .216; .221, .066]$ \\
\netw{GNUTELLA-25} &  $22687$ & $2.41$ & $[.351, .233; .308, .086]$ \\
\netw{GNUTELLA-30} &  $36682$ & $2.41$ & $[.355, .231; .298, .084]$ \\
\netw{WEB-NOTREDAME} &  $325729$ & $4.60$ & $[.460, .190; .208, .105]$ \\
  \bottomrule
  \end{tabular}
  \caption{%
    Initiator matrices of the 20 networks fitted to the \SKG model in~\cite{leskovec2010kronecker}, converted
    to \RMAT initiator matrices. The upper half of the networks are undirected and hence have a symmetric initiator.
    We see that the densities are consistently small which supports the common assumption that~$\mu$ is a constant independent
    of the graph size.
  \label{tab:rmat-converted}}
\end{table}

\subsection{Analytical tools}

\noindent
We use the following tools for asymptotic analysis:\\

\noindent \textbf{Hamming slices:}
Since the bitstring representation of the vertices in Kronecker-style models
encodes information about the edges between them, it will be useful to group
the vertices by properties of their bitstrings. The \emph{Hamming weight} of a
vertex is the number of ones in its bitstring label. We define a \emph{Hamming
slice}~$\mathcal F_\ell$ to be the set of all vertices whose bitstrings have
Hamming weight exactly~$\ell$. We also denote $\mathcal F_{\leq \ell}$ and
$\mathcal F_{\geq \ell}$ as the set of vertices with bitstrings at most and at
least Hamming weight~$\ell$, respectively.\\

\noindent \textbf{Asymptotic equivalence:}
Two random graph models~$(\P_n),(\Q_n)$ are \emph{asymptotically equivalent} if for every
sequence of events~$(\mathcal E_n \mid \mathcal E_n \in 2^{\mathcal G_n} )_{n
\in \N}$ it holds that
\[
  \lim_{n \to \infty} \P_n[\mathcal E_n] - \Q_n[\mathcal E_n] = 0.
\]

\noindent \textbf{Concentration inequalities:}
To show that the invariants of our graphs do not deviate far from their expected values, we use \emph{Chernoff--Hoeffding bounds}.
\begin{ChernoffHoeffding}[\cite{hoeffding1963probability}]
  Let~$\{X_i\}_{i \in [n]}$ be random binary variables with associated success probabilities
  $\{p_i\}_{i \in [n]}$. Let further~$\xi = \E[ \sum_i X_i ]$. Then for every~$\delta > 0$ it holds that
  \[
    \Pr\left[\sum_i X_i \geq (1+\delta) \xi \right] \leq \left( \frac{e^\delta}{(1+\delta)^{1+\delta}} \right)^\xi.
  \]
\end{ChernoffHoeffding}

\noindent
A common situation will be that~$\xi$ tends towards zero as~$n$ increases.
Choosing some constant~$c$ for which we want to obtain a bound, we let~$c =
(1+\delta) \xi$ from which we derive that $\delta \xi = c -\xi$ and~$1+\delta
= c/\xi$. We thus use the following reformulation of this bound:
\begin{align*}
  \Pr\left[\sum_i X_i \geq c \right] \leq \frac{e^{\delta\xi}}{(1+\delta)^{(1+\delta)\xi}}
        = \frac{(e\xi)^c}{e^\xi c^c } \leq \left( \frac{e \xi}{c} \right)^c.
\end{align*}

\noindent \textbf{Binomial Coefficients}: We also use the following bound on
binomial coefficients based on Stirling's approximation:
\begin{align*}
  \frac{\sqrt \pi}{2} \Gamma(\tau) \leq {n \choose \tau n} \leq \Gamma(\tau)
  ~\text{where}~ \Gamma(\tau) = \frac{2^{nH(\tau)}}{\sqrt{2\pi n \tau(1{-}\tau)}},
  \label{eqn:binomial_bound}
\end{align*}
and~$H(x) = -\log_2(x^x(1-x)^{1-x})$ is the \emph{binary entropy function}.\\

\noindent \textbf{Degeneracy, cores, and dense subgraphs:}
Recent work on community structure in complex networks has pointed 
to some sort of ``core-periphery'' structure in many real
networks (\eg~\cite{leskovec2009community, adcock2013, rombach2014}), often 
exemplified using the $k${\emph-core decomposition}, a popular 
tool in visualization and social network analysis (see \eg~\cite{giatsidis2011evaluating, alvarez2005large, kitsak2010identification, carmi2007model, alvarez2008k-coredecomposition}). 

The \emph{$k$-core} of a graph is the
maximal induced subgraph in which all vertices have degree at least
$k$. The core-periphery structure of a network is often characterized in terms of
the depth of its core-decomposition (largest $k$ such that the $k$-core is non-empty), 
an invariant known as the \emph{degeneracy}. Both the degeneracy of a graph and its core
decomposition can be computed in $O(m)$ time using an algorithm by Batagelj and
Zaversnik \cite{batagelj2003m}. 
In Section~\ref{sec:degeneracy}, we analyze the degeneracy of Kronecker-style models, 
using the following equivalent definition where needed: 

\begin{definition}[Degeneracy]
  A graph $G$ is said to be \emph{$d$-degenerate} if every (induced) subgraph
  has a vertex of degree at most $d$. The smallest $d$ for which $G$ is
  $d$-degenerate is the \emph{degeneracy} of $G$.
\end{definition}

\noindent
In particular, all subgraphs of a $d$-degenerate graph are sparse; in the 
contrapositive this means that the existence of a subgraph of density~$> d$
implies that the host graph is \emph{not} $d$-degenerate. In
the asymptotic setting, a dense subgraph is a sequence of subgraphs whose edge
density diverges in the limit; the existence of such a substructure implies that
for every integer~$d$ the generated graphs past a certain threshold size 
are not $d$-degenerate. We will call graph models that contain dense subgraphs
\emph{asymptotically dense}. \looseness-1

\noindent

\subsection{Repeatability}

All experiments in this paper can be replicated with the code available at
\url{http://dl.dropboxusercontent.com/s/vfzvk72gpqbmhc9/asymptotic_kronecker.zip}. This
includes random graph model implementations, the random seeds used to
generate the data used in this manuscript, and code to calculate and plot the
relevant graph invariants. All code is written in Python; we recommend running
with \texttt{pypy} to reduce runtimes.

%
\section{Relationships of Kronecker-style Models}\label{sec:model-equiv}
We begin this section by proving that all of the aforementioned variants of \RMAT are equivalent asympotically. We then show that equivalent input parameters to \RMATFlip and \SKG do not generate equal probability distributions over the arcs. To further validate this proof, we provide an empirical result highlighting the difference between the models.

\begin{figure*}[!t]
  \centering
  \includegraphics[width=7in]{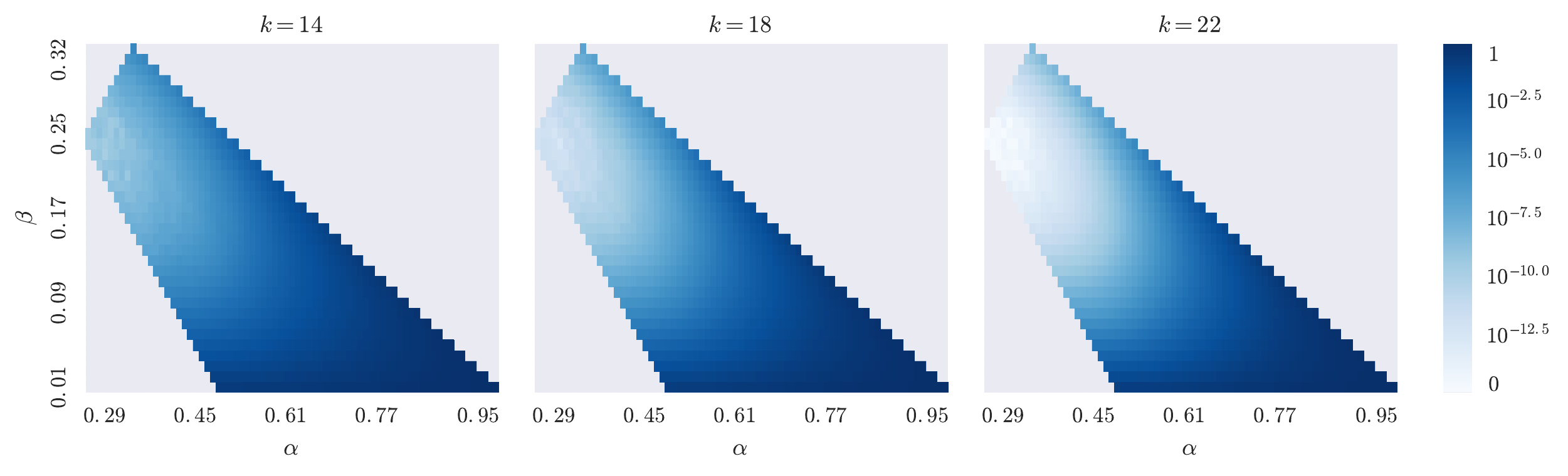}
  \caption{Proportion of arcs requiring a rethrow in \RMATRethrow as a function of $k$, $\alpha$, and $\beta$, averaged over ten graphs for each parameter value. All graphs were generated with $\mu = 4$ and $\gamma = \beta$, fixing $\delta = 1-\alpha-2\beta$.}\label{figure:collisions}
\end{figure*}

\subsection{Equivalence of \RMAT variants}
\noindent
In this section, we show the following equivalence between the three \RMAT variants:

\begin{theorem}\label{theorem:erasure-rethrow-equiv}
  For parameters $\twobytwo(\alpha, \beta, \gamma, \delta)$ where at least
  three entries are non-zero
  and~$\mu \leq \log(n)/2$ the models~$\RMATErasure$,  $\RMATRethrow$ and $\RMATFlip$
  are asymptotically equivalent for appropriate scalings of the parameter~$\mu$.
\end{theorem}

To prove this, we first establish that---in the sparse case---the erasure and rethrow models
result asymptotically in the same process (Lemma~\ref{lemma:erasure-rethrow-equiv}, which uses
Lemmas~\ref{lemma:occupied-weight} and~\ref{lemma:erasures} to bound the number of rethrows).
Equivalence with the coin-flipping model then follows easily in Lemma~\ref{lemma:erasure-flip-equiv}
using the edge probabilities defined in the general \RMAT process.

We start by estimating the probability that an arc lands on an occupied cell and
therefore is handled differently in the erasure and rethrow models.
In the remainder of this section, we fix an initiator matrix~$\twobytwo(\alpha, \beta, \gamma, \delta)$
and let ~$\rho_1 \geq \rho_2 \geq \rho_3 \geq \rho_4$ be its entries ordered by size. As before we denote by~$k$ the number of
Kronecker-multiplications and by~$\mu$ the density parameter. The following lemma
holds even when for a superlinear number of arcs and we state it in that generality,
however, our subsequent application will again assume a linear number of arcs.

\begin{lemma}\label{lemma:occupied-weight}
  Let~$\dir G$ be a graph on~$n = 2^k$
  vertices with~$cn$ arcs for~$1 \leq c \leq k / 2$. There exists a
  function~$f$ depending on~$\rho_1,\rho_2$ and~$\rho_3$ such that
  \[
    \sum_{e \in \dir G} \omega_{e} \leq (\rho_1+\rho_2)^k \left( \rho_3 + f(c) \right).
  \]
\end{lemma}
\begin{proof}
  The number of weights that only consist of the largest factors~$\rho_1$
  and~$\rho_2$ is exactly~$2^k = n$. For~$c = 1$, we have
  that at most a total weight of
  \[
    \sum_{i=0}^k {k \choose i} \rho_1^i \rho_2^{k-i} = (\rho_1+\rho_2)^k
  \]
  is occupied by arcs in~$\dir G$. We take these highest weights and replace up
  to~$g$ positions with~$\rho_3$ in order to increase maximum weight covered by~$cn$ arcs: this
  provides us with at least
  \begin{align*}
    &\phantom{{}={}}
      \sum_{j=0}^g {k \choose j} \sum_{i=0}^{k-j} {k-j \choose i} = \sum_{j=0}^g {k \choose j} 2^{k-j} \\
    &\geq 2^k + 2^k \sum_{j=1}^g \left( \frac{k}{2j} \right)^j \geq 2^k \left(1 + \frac{gk}{2} \right).
  \end{align*}
  arcs whose weights only consist of factors~$\rho_1,\rho_2$ and up to~$g$ factors~$\rho_3$. In order
  to now bound the total \emph{weight} such arcs can occupy, we solve for~$g$:
  \begin{align*}
    & 1 + \frac{gk}{2} = c
    \iff g = 2\frac{c-1}{k}.
  \end{align*}
  The total weight of the~$cn$ arcs in~$\dir G$ is therefore at most
  \begin{align*}
    &\phantom{{}\leq{}}
      \sum_{j=0}^g {k \choose j} \rho_3^j \sum_{i=0}^{k-j}
                 {k-j \choose i} \rho_1^i \rho_2^{k-j-i}\\
    &= \sum_{j=0}^g {k \choose j} \rho_3^j (\rho_1+\rho_2)^{k-j} \\
    &\leq (\rho_1+\rho_2)^k \left( \rho_3
       + \sum_{j=1}^g \left(
         \frac{ek\rho_3}{j(\rho_1+\rho_2)} \right)^j
                      \right)\\
    &\leq   (\rho_1+\rho_2)^k \left( \rho_3
      + 2\left(\frac{e\rho_3}{\rho_1+\rho_2} k\right)^g - 1\right).
  \end{align*}
  With the above value for~$g$, we bound the inner term by
  \begin{align*}
    \left(\frac{e\rho_3}{\rho_1+\rho_2} k \right)^g =
    \left( \frac{e\rho_3}{\rho_1+\rho_2} k \right)^{\frac{1}{k} \cdot 2(c-1)}
    \leq e^{ \frac{2}{\rho_1+\rho_2} \rho_3 (c-1) },
  \end{align*}
  where we used the fact that~$(\xi k)^{1/k}$ achieves its
  maximum at~$k=e/\xi$ for~$k > 0$. We conclude that the total
  weight of~$cn$ occupied arcs is at most
  \begin{align*}
    (\rho_1+\rho_2)^k \left( \rho_3 + f(c) \right),
  \end{align*}
  as claimed.
\end{proof}

\noindent
A direct consequence of Lemma~\ref{lemma:occupied-weight} is that the
probability that the~$i^\text{th}$ arc in the \RMAT process will hit an occupied
arc is at most
\[
  \P[e_i \in \dir G^{i-1}] \leq (\rho_1+\rho_2)^k (\rho_3+f(c)).
\]
We use this result to calculate the order of the number of expected collisions in \RMATErasure.
\begin{lemma}\label{lemma:erasures}
  The expected number of erased arcs
  in~$\RMATErasure\big(\twobytwo(\alpha,\beta,\gamma,\delta),\mu,k\big)$ is
  \[
    (\rho_1+\rho_2)^k (\rho_3+f(\mu)) \cdot \mu n = O( (\rho_1+\rho_2)^k n)
  \]
  with high probability.
\end{lemma}
\begin{proof}
  Let~$m = \mu n$ and consider the sequence~$(\dir G_i)_1^m$ of graphs generated
  by the model. Let~$(X_i)_1^m$ be a sequence of random binary variables where
  $X_i = 1$ iff the $i^\text{th}$ arc was erased. By Lemma~\ref{lemma:occupied-weight},
  we have that
  \[
    \P[X_i = 1] \leq \P[X_m = 1] \leq (\rho_1+\rho_2)^k (\rho_3+f(\mu)).
  \]
  The expected number of erased arcs is therefore
  \begin{align*}
    \E\left[\, \sum_{i=1}^m X_i \,\right] &\leq m \cdot \E[X_m]
      = m(\rho_1+\rho_2)^k (\rho_3+f(\mu))\\
      &= O( 2^k(\rho_1+\rho_2)^k),
  \end{align*}
  and by the usual concentration arguments the actual value is bounded by
  this quantity with high probability.
\end{proof}

\noindent
For example, with~$\rho_1+\rho_2 \leq 1/2$ we expect that only a constant
number of arcs will be erased and with~$\rho_1 + \rho_2 \leq 1/\sqrt{2}$ we
expect $O(\sqrt n)$ erasures (\cf~Figure \ref{figure:collisions}). For a real-world example,
note that Lemma~\ref{lemma:occupied-weight} guarantees for almost
all parameters listed in Table~\ref{tab:rmat-converted} that the probability
of an arc being erased given the size and existing number of arcs in the network is below~$5\%$. Notable exceptions
are the relatively small and dense networks~\netw{Ca-HEP-Ph}
and~\netw{Email-Inside} for which the bound fails to give meaningful values.

We now use Lemma~\ref{lemma:erasures} to prove the asymptotic equivalence of \RMATRethrow and \RMATErasure.

\begin{lemma}\label{lemma:erasure-rethrow-equiv}
  If~$\rho_1+\rho_2 < 1$, then we can couple the generation of
  $\dir G^\oplus = \RMATRethrow\big(\twobytwo(\alpha,\beta,\gamma,\delta),\mu,k\big)$
  with $\dir G^\ominus = \RMATErasure\big(\twobytwo(\alpha,\beta,\gamma,\delta),\mu',k\big)$ where
  $\mu' \sim \mu$ such that~$\dir G^\oplus = \dir G^\ominus$ with high probability.
\end{lemma}
\begin{proof}
  We first generate the sequence~$(\dir G^\ominus_i)_1^{m'}$ with~$m' = \mu' n$ using the
  erasure model. By Lemma~\ref{lemma:erasures}, the resulting graph~$G^\ominus_{m'}$
  will have, with high probability,
  \[
    \left(1 - (\rho_1+\rho_2)^k (\rho_3+f(\mu')) \right) \mu' n
    = \left(1-o(1)\right) \mu' n \sim \mu n
  \]
  arcs. To generate~$(\dir G^\oplus_i)_1^m$ we simply add the first~$\mu n$ arcs
  in the same order, reinterpreting an erasure as a rethrow.
\end{proof}

\noindent
Using the probability of an arc's existence in the \RMAT process 
(\cf~Preliminaries), we can similarly relate \RMATErasure and \RMATFlip.

\begin{lemma}\label{lemma:erasure-flip-equiv}
  If~$\rho_1+\rho_2 < 1$,
  then we can couple the generation of
  $\dir G^{\$} = \RMATFlip\big(\twobytwo(\alpha,\beta,\gamma,\delta),\mu,k\big)$
  with $\dir G^\ominus = \RMATErasure\big(\twobytwo(\alpha,\beta,\gamma,\delta),\mu,k\big)$
  such that~$\dir G^{\$} = \dir G^\ominus$.
\end{lemma}
\begin{proof}
  We first generate the sequence~$(\dir G^\ominus_i)_1^{m}$ with~$m = \mu n$ using the
  erasure model. As observed above, the probability that an arc~$e$ is present in the
  graph is given by
  \[
    \P[e \in \dir G^\ominus] = 1 - (1-\omega_e)^m,
  \]
  which is exactly the probability that the arc is contained in~$\dir G^{\$}$.
\end{proof}

\subsection{Differences between \RMAT and \SKG}

\noindent
We naturally ask ourselves whether Theorem~\ref{theorem:erasure-rethrow-equiv}
is true for \SKG and \RMAT. Our primary observation will be that while most
arc-probabilities in the models converge, the speed of this convergence depends
on the respective arc-weights.
To demonstrate this skew, let us first introduce the following
variant on Bernoulli's inequality. For brevity's sake, we use the
symbol~$\gtreqless$ which indicates that the following lemma is true if all
appearances of~$\gtreqless$ are simultaneously replaced by either~$\leq$
or~$\geq$.

\begin{lemma}\label{lemma:bernoulli}
  For every function~$f\colon \R \to \R^+$ and integer~$t \geq 1$ with
  $f(1) \gtreqless 1$ it holds that
  \[
    (1-x)^t \gtreqless 1 - f(t) x
  \]
  for every~$x \in (0,1]$ with~$x \gtreqless 1 - \frac{f(t)-1}{f(t-1)}$.
\end{lemma}
\begin{proof}
  We use induction over~$t$. Since~$f(1) \gtreqless 1$, we have
  that~$1-x \gtreqless 1 - f(1) x$ and hence the basis for induction.
  Then it follows that
  \begin{align*}
    &\phantom{{}={}} (1-x)^t \\
    &= (1-x)^{t-1} (1-x)\\
    &\gtreqless (1-f(t-1)x) (1-x) \\
    &= 1 - f(t)x + x\left(f(t)-f(t-1)-1+f(t-1)x \right).
  \end{align*}
  Since~$x > 0$, the bound follows when
  \begin{alignat*}{3}
   f(t) - f(t-1) - 1 + f(t-1)x &\gtreqless 0 \\
   f(t-1)x &\gtreqless f(t-1) - f(t) + 1 \\
   x &\gtreqless 1 - \frac{f(t) - 1}{f(t-1)},
  \end{alignat*}
  as claimed.
\end{proof}

\noindent
For~$f(t) = t$ and~$x \geq 0$ we simply recover the Bernoulli-bound
\[
  (1-x)^t \geq 1-tx.
\]
Since equality is reached exactly whenever
\[
  f(t) = \frac{1-(1-x)^t}{x},
\]
we see that the approximation~$f(t) = t$ is best for very small~$x$. Further,
we have the following asymptotic relationships for particular dependencies of~$x$ and~$t$:
\begin{center}
  \begin{tabular}{ll>{\hspace*{-9pt}}l>{\hspace*{-7pt}}l}
    $x = t^{-1}$   &$\implies$   & $(1-(1-t^{-1})^t) \cdot t$ & $\sim t - \frac{1}{e}t$ \\
    $x = t^{-1.5}$ &$\implies$   & $(1-(1-t^{-1.5})^t) \cdot t^{1.5}$ & $\sim t - \sqrt{t}$ \\
    $x = t^{-2}$   &$\implies$   & $(1-(1-t^{-2})^t) \cdot t^2$ & $\sim t - \log t$
  \end{tabular}
\end{center}
Relating these asymptotic relations to \SKG and \RMAT probabilities,
the above suggests for weights~$\omega_e = \Theta(n^{-2})$ that the binomial arc-probability
in \RMAT models is best approximated by
\[
  1 - (1 - \omega_e)^m \sim (m - \log m) \cdot \omega_e,
\]
which is reasonably close to the corresponding arc probability~$m \omega_e$ in the \SKG model.
For an arc weight of~$\omega_e = \Theta(n^{-1})$, however, we have that
\[
  1 - (1 - \omega_e)^m \sim \frac{e-1}{e} m \cdot \omega_e \approx \frac{2}{3} m \cdot \omega_e.
\]
Hence arcs with large weight will have significantly different probabilities in
\SKG compared to the \RMAT models. Since this difference is inhomogeneous in all interesting
cases it cannot be remedied by simple scaling of probabilities.

On the positive side, most arc probabilities are reasonably similar in both models and
we can prove a weaker kind of equivalence between \RMAT and \SKG. The following
relationship between \RMATFlip and \SKG extends via Theorem~\ref{theorem:erasure-rethrow-equiv}
to the other two \RMAT variants. Note that the factor~$2/3$ in the following is
chosen for convenience and can be replaced by any fixed number in~$(0,1)$.

\begin{theorem}\label{theorem:SKG-sandwich}
  Assume that~$\rho_1 < 1/2$. Let~$\theta =  2 \mu^{1/k}$.
  We can couple the generation of
  $G^{\$} = \RMATFlip\big(\twobytwo(\alpha,\beta,\gamma,\delta),\mu,k\big)$
  with $G_1^{\otimes} = \SKG\big(\theta\twobytwo(\alpha,\beta,\gamma,\delta),\frac{2}{3}\mu,k\big)$
  and~$G_2^{\otimes} = \SKG\big(\theta\twobytwo(\alpha,\beta,\gamma,\delta),\mu,k\big)$ such that
  $G_1^{\otimes} \subseteq G^{\$} \subseteq G_2^{\otimes}$.
\end{theorem}
\begin{proof}
  We apply Lemma~\ref{lemma:bernoulli} using~$f(t) = \frac{2}{3} t$ and
  obtain that
  \[
    (1-x)^t \leq 1 - \frac{2}{3}t x
  \]
  whenever~$x \leq \frac{1}{2(t-1)}$. Let~$m = \mu n$ as usual. Because~$\rho_1 < 1/2$,
  eventually the inequality
  \[
    \rho_1^k \leq \frac{1}{2(m - 1)} = \frac{1}{2(\mu \cdot 2^k - 1)}
  \]
  holds. For every arc~$e$ we then have that
  \[
    \frac{2}{3} m \cdot \omega_e \leq 1 - (1 - \omega_e)^m \leq m \cdot \omega_e.
  \]
  Note that the upper and lower bounds are exactly the respective
  probabilities for the arc~$e$ in~$G_1^{\otimes}$ and~$G_2^{\otimes}$.
  Accordingly,
  \[
    \P[e \in G_1^{\otimes}] \leq \P[e \in G^{\$}] \leq \P[e \in G_2^{\otimes}]
  \]
  and the coupling is straightforward.
\end{proof}


\noindent
To supplement the theoretical claim that ``equivalent'' input parameters to
\SKG and \RMATFlip generate unequal probability distributions over the arcs,
we would like to demonstrate that this probability difference translates into
observable differences between the two models. We proved that the biggest
discrepancy between edge probabilities occurs in arcs with the largest
weights, which roughly translates to arcs in the lowest Hamming
slices\footnote{This effect is strongest when $\alpha, \beta$ are much
larger than~$\gamma$ and~$\delta$.}. Figure~\ref{figure:edge_count_diff} shows that the
average number of arcs in the lowest Hamming slices is noticeably larger when
using the \RMATFlip model. Since this (somewhat artificial) statistic can differentiate between the
two models the common approach of using \RMAT and \SKG interchangeably needs to be
scrutinized. The question of whether other, more natural statistics
diverge on these graph models will be interesting for future research.

\begin{figure}
  \centering
  \includegraphics[width=3.75in]{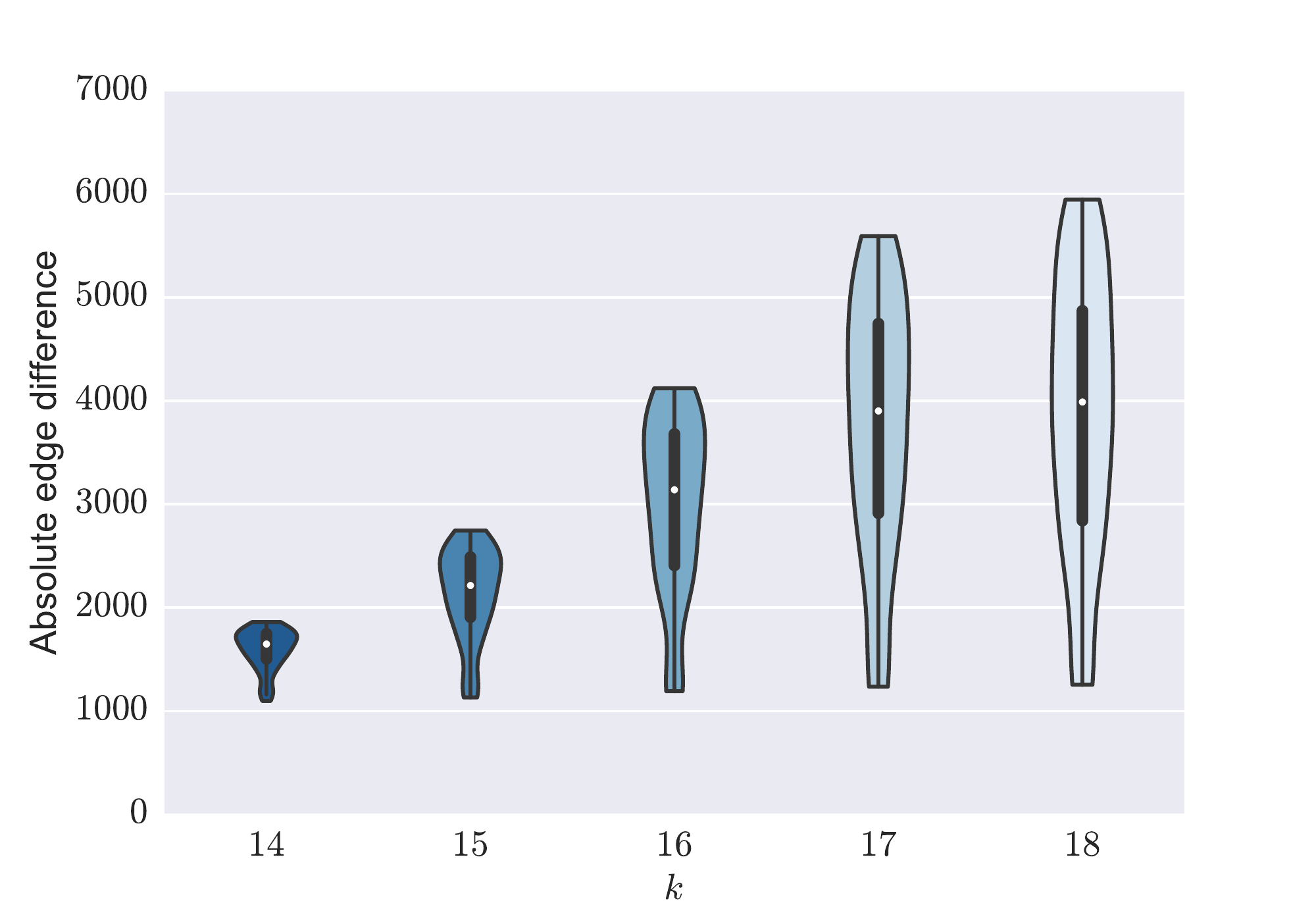}
  \caption{Distribution of the difference between number of edges in the
           subgraph induced on the $|\mathcal F_{\leq 6}|$ vertices of highest
           degree in graphs generated by \RMATFlip and \SKG. Ten graphs were
           generated using parameters $\mu = 6$ and $(0.45,0.275,0.275,0)$
           for each size.}\label{figure:edge_count_diff}
\end{figure}

%
\section{Degeneracy of Kronecker-style models }\label{sec:degeneracy}

\noindent
We now proceed to an analysis of the degeneracy of Kronecker-style models.
In keeping with prior results about degeneracy in random graph
models~\cite{farrell2015hyperbolicity,fernholz2003giant,pittel1996sudden}, we
will analyze how the degeneracy changes as we grow the graph size while
holding the average degree constant.
Specifically, we want to resolve whether the degeneracy is \emph{bounded}
(converges to a constant) or \emph{unbounded} (grows arbitrarily large
with the size of the graph).

Since degeneracy is a property of undirected graphs we will only consider \emph{symmetric}
initiator matrices of the form~\twobytwo(\alpha,\beta,\beta,\delta) in the following. The
edge~$uv$ is then present in the final graph if at least one of the arcs~$uv,vu$ is contained
in the generated digraph.
Due to Theorem~\ref{theorem:SKG-sandwich} we can translate any results for \SKG that holds independently
of the value of~$\mu$ (as long as $\mu$ does not scale with~$k$) to the \RMAT models and vice versa.
We focus on \SKG here since it is much easier to analyze.  Note that in the symmetric case we have
that the \emph{edge}~$uv$ is added with probability~$2m \omega_{uv}$.

The following expression will be crucial in all following calculations.
Fix a constant~$0 < \tau  < 1$ and consider a vertex~$x \in \mathcal
F_{\tau k}$. The expected number of edges from~$x$ to vertices in~$\mathcal
F_{\leq \tau k}$ (including a loop to itself) is then given by
\begin{align*}
  &\E[\deg_{\leq \tau k}(x)]\\
  &= 2m \sum_{\ell \leq \tau k} \sum_{i \leq \ell}
      {\tau k \choose i} {(1-\tau) k \choose \ell -i} \delta^i \beta^{\tau k + \ell - 2i} \alpha^{(1-\tau)k-\ell+i} \\
  &= 2m \alpha^{(1-\tau)k}\beta^{\tau k}
    \! \sum_{i \leq \ell \leq \tau k} \! {\tau k \choose i} {(1-\tau) k \choose \ell -i}
                          \Big(\frac{\delta}{\beta}\Big)^i \Big(\frac{\beta}{\alpha}\Big)^{\ell-i} \\
  &=: 2\mu \cdot (2\alpha)^k (\beta/\alpha)^{\tau k} \cdot \Lambda_\tau(k).
      \addtocounter{equation}{1}\tag{\theequation}\label{eq:degree}
\end{align*}
We derive the following lower bound on~$\Lambda_\tau$:

\begin{lemma}\label{lemma:sum-lower-bound}
  Assuming that~$\beta/\alpha - \delta/\beta \leq 1/e$,
  it holds that
  \begin{equation*}
      \Lambda_\tau(k) \geq \begin{cases}
          e^{(\tau\delta/e\beta + (1-\tau)\beta/e\alpha)k} & \text{for}~\tau \geq \frac{1}{1 + \frac{e\alpha}{\beta} - \frac{\alpha\delta}{\beta^2}} \\
          (\delta/\beta+\frac{1-\tau}{\tau} \beta/\alpha)^{\tau k} & \text{otherwise}.
      \end{cases}
  \end{equation*}
  For~$\beta/\alpha - \delta/\beta > 1/e$ the two bounds swap.
\end{lemma}
\begin{proof}
  By applying the bound~${n \choose k} \geq (n/k)^k$ to both binomial coefficients,
  we obtain that
  \begin{align*}
    \Lambda_\tau(k) &\geq \sum_{\ell \leq \tau k} \! k^\ell
                  \sum_{i \leq \ell}
                      \Big(\frac{\tau \delta/\beta}{i}\Big)^i
                      \Big(\frac{(1-\tau)\beta/\alpha}{\ell-i}\Big)^{\ell-i} \\
                    &\geq \tau k \sum_{\ell \leq \tau k}
                      \Big(\frac{k}{\ell}\Big)^\ell (\tau\delta/\beta + (1-\tau)\beta/\alpha)^\ell.
  \end{align*}
  For $\tau \geq (1 + e\alpha/\beta - \alpha\delta/\beta^2)^{-1}$, this sum's largest
  term occurs at~$\ell = e(\tau\delta/\beta + (1-\tau)\beta/\alpha) k$ with a value of
  $
    e^{(\tau\delta/e\beta + (1-\tau)\beta/e\alpha)k}.
  $
  For $\tau < (1 + e\alpha/\beta - \alpha\delta/\beta^2)^{-1}$, the sum achieves its maximum at~$\tau k$
  with a value of
  $
    (\delta/\beta+\frac{1-\tau}{\tau} \beta/\alpha)^{\tau k}
  $,
  from which the second bound follows. The inequalities with respect to~$\tau$ invert
  whenever~$\beta/\alpha - \delta/\beta > 1/e$.
\end{proof}

\begin{corollary}\label{cor:deg-lower-bound}
  Assume that~$\beta/\alpha - \delta/\beta \leq 1/e$ and fix~$0 < \tau < 1$.
  For every vertex~$x \in \mathcal F_{\tau k}$ it holds that
  \begin{align*}
      \E[\deg_{\leq \tau k}(x)] &\geq k^{\Theta(1)} \left( 2 \alpha (\beta/\alpha)^\tau e^{(\tau\delta/e\beta + (1-\tau)\beta/e\alpha)} \right)^k
  \intertext{%
  for~$\tau \geq \frac{1}{1 + \frac{e\alpha}{\beta} - \frac{\alpha\delta}{\beta^2}}$ and
  }
      \E[\deg_{\leq \tau k}(x)] &\geq k^{\Theta(1)} \left( 2 \alpha (\delta/\alpha+\frac{1-\tau}{\tau} \beta^2/\alpha^2)^\tau \right)^k
  \end{align*}
  otherwise. For~$\beta/\alpha - \delta/\beta > 1/e$ the two bounds swap.
\end{corollary}

\noindent The above bounds are general, but suffer from the usual shortcomings of approximating
binomial coefficients by simpler functions. The following bounds are geared towards special parametric
ranges and can be taken in conjunction to obtain a more complete picture of the parameter space:
\begin{corollary}\label{cor:simple-deg-lower-bound}
  For every~$0 < \epsilon < 1$ and~$\tau \leq \min\{\epsilon,1/2\}$ it holds that
  \begin{align*}
    \E[\deg_{\leq \tau k}(x)] &\geq k^{\Theta(1)} \Big( 2\alpha (\beta/\alpha)^{(1+\epsilon)\tau} 2^{H(\epsilon\tau/(1-\tau))(1-\tau)} \Big)^k. \\
  \end{align*}
\end{corollary}
\begin{proof}
  We set~$i = 0$ and~$\ell = \epsilon\tau k$ in Equation~\ref{eq:degree} and bound the remaining binomial
  coefficient by~${(1-\tau)k \choose \epsilon\tau k} \leq k^{\Theta(1)} 2^{H(\epsilon\tau/(1-\tau)) (1-\tau) k}$.
\end{proof}

\noindent
We can prove an analog to Lemma~\ref{lemma:sum-lower-bound} and Corollary~\ref{cor:deg-lower-bound} to obtain upper bounds using the same techniques.
We omit the proof here, note that the bounds on~$\tau$ and~$\Lambda_\tau$ differ slightly.

\begin{lemma}\label{lemma:sum-upper-bound}
  Assuming that~$\beta/\alpha - \delta/\beta \leq 1$,
  it holds that
  \begin{equation*}
      \Lambda_\tau(k) \leq \begin{cases}
          \tau k e^{(\tau\delta/\beta + (1-\tau)\beta/\alpha)k} & \text{for}~\tau \geq \frac{1}{1 + \frac{\alpha}{\beta} - \frac{\alpha\delta}{\beta^2}} \\
          \tau k (e\delta/\beta+\frac{1-\tau}{\tau} e\beta/\alpha)^{\tau k} & \text{otherwise}.
      \end{cases}
  \end{equation*}
  For~$\beta/\alpha - \delta/\beta > 1$ the two bounds swap.
\end{lemma}

\begin{corollary}\label{cor:deg-upper-bound}
  Assume that~$\beta/\alpha - \delta/\beta \leq 1$ and fix~$0 < \tau < 1$.
  For every vertex~$x \in \mathcal F_{\tau k}$ it holds that
  \begin{align*}
      \E[\deg_{\leq \tau k}(x)] &\leq k^{\Theta(1)} \left( 2 \alpha (\beta/\alpha)^\tau e^{(\tau\delta/\beta + (1-\tau)\beta/\alpha)} \right)^k
  \intertext{%
  for~$\tau \geq \frac{1}{1 + \frac{\alpha}{\beta} - \frac{\alpha\delta}{\beta^2}}$ and
  }
      \E[\deg_{\leq \tau k}(x)] &\leq k^{\Theta(1)} \left( 2 \alpha (e\delta/\alpha+\frac{1-\tau}{\tau} e\beta^2/\alpha^2)^\tau \right)^k
  \end{align*}
  otherwise. For~$\beta/\alpha - \delta/\beta > 1$ the two bounds swap.
\end{corollary}

\noindent
We will now relate the quantity~$\E[\deg_{\leq \tau k}(\cdot)]$ to the existence of dense
and sparse subgraphs and apply the derived lower and upper bounds to identify parametric
ranges in which these structures are asymptotically unavoidable.

%
%
\subsection{Lower Hamming slices are dense}\label{sec:degen-dense}

\noindent
Let~$G$ be the undirected graph generated by~\SKG with the
parameters~$\twobytwo(\alpha,\beta,\beta,\delta),\mu$ and~$k$. An simple first
observation with respect to the density of \RMAT models
is that for~$\alpha \geq 1/2$, already the Hamming slice~$\mathcal F_1$ is asymptotically 
dense: the expected density of ~$\mathcal F_1$ is
\[
  2m k^{-1} {k \choose 2} \alpha^{k-1}\beta \sim 2\mu k (2\alpha)^k (\beta/\alpha),
\]
which goes to infinity as~$k$ grows and hence produces a dense subgraph for~$\alpha \geq 1/2$.
We want to extend this observation and ask for the range~$\alpha < 1/2$ whether the
lower Hamming slices are asymptotically dense.

By some abuse
of notation, let us write~$\|F_{\leq k}\| := |E(G[\mathcal F_{\leq k}])|$ to
denote the number of edges whose endpoints both have Hamming weight at
most~$k$ in~$G$. We define the density~$D_{\leq \tau k} = \|\mathcal F_{\leq
\tau k}\| / |\mathcal F_{\leq \tau k}|$. Let us first relate this density to
the expected number of neighbors a vertex has in lower Hamming slices.

\begin{figure}[bth]
  \centering
  \hspace*{-6pt}
  \includegraphics[scale=.6]{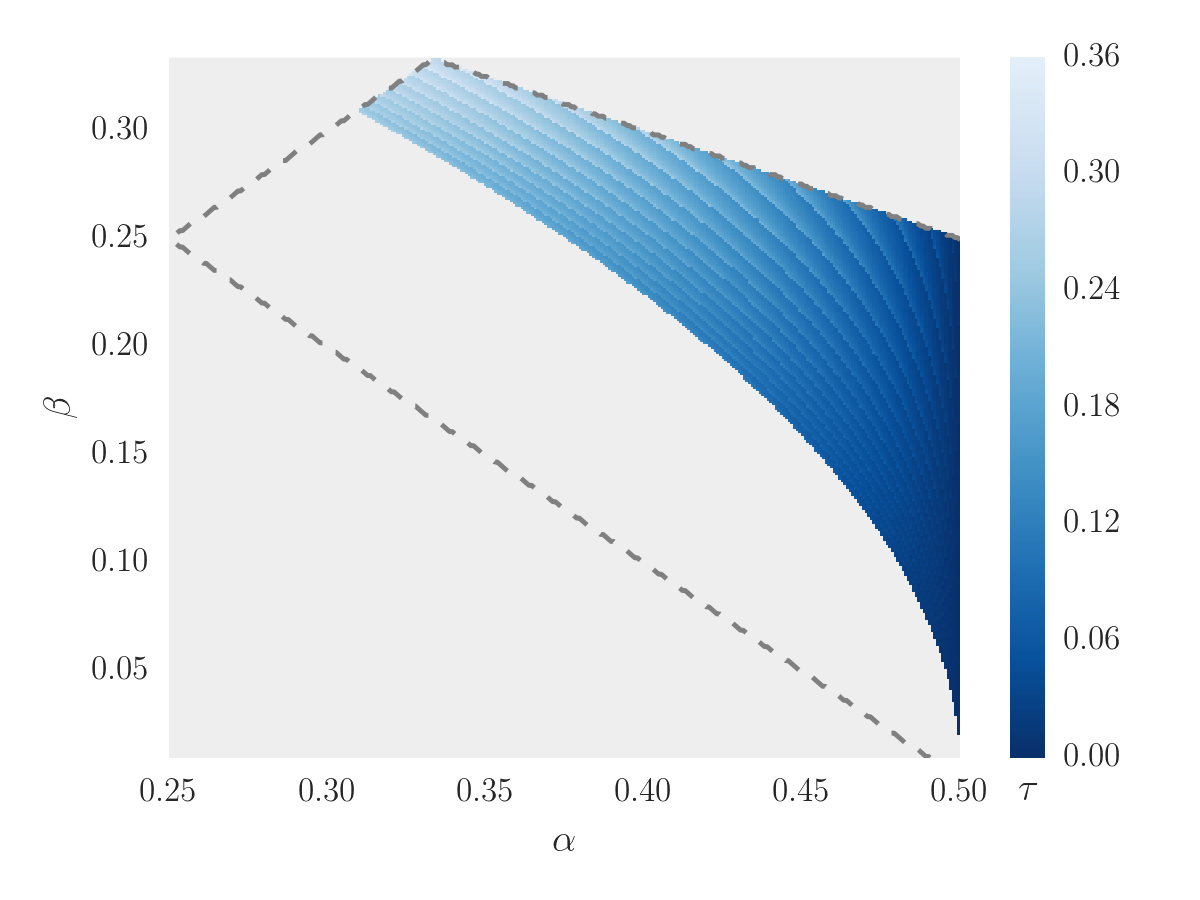}
  \caption{%
    Parametric range in which Kronecker-style models generate dense subgraphs according to the bound
    derived in Corollary~\ref{cor:simple-deg-lower-bound}. The range in which~$\alpha$ is the largest
    entry of the initiator matrix is circumscribed by the dashed outline.
    The shading indicates the smallest~$\tau$ for which the density~$D_{\leq \tau k}$ is above one according
    to the lower bound, gray regions indicate that it lies below one. The banding structure is an artifact
    of using 20 equally distributed values for~$\epsilon$.
    \label{fig:density}
  }
\end{figure}

\begin{lemma}\label{lemma:density-bound}
  For every~$0 < \tau < 1$ it holds that
  \[
    D_{\leq \tau k} \geq k^{\Theta(1)} \E[\deg_{\tau k}(x) \mid x \in \mathcal F_{\tau k}]
  \]
  with high probability.
\end{lemma}
\begin{proof}
  We have that
  \begin{align*}
     D_{\leq \tau k} &= \frac{\|\mathcal F_{\leq \tau k}\|}{|\mathcal F_{\leq \tau k}|}
    \geq \frac{|\mathcal F_{\tau k}|}{|\mathcal F_{\leq \tau k}|} \E[\deg_{\leq \tau k}(x) \mid x \in \mathcal F_{\tau k}] \\
    &= {k \choose \tau k} \Big(  \sum_{i \leq \tau k} {k \choose i} \Big)^{-1} \E[\deg_{\leq \tau k}(x) \mid x \in \mathcal F_{\tau k}] \\
    &= k^{\Theta(1)} \E[\deg_{\leq \tau k}(x) \mid x \in \mathcal F_{\tau k}],
  \end{align*}
  and the claim follows from concentration arguments.
\end{proof}

\noindent
It follows that~$\E[\deg_{\tau k}(\cdot)]$ is crucial for the density of graph
generated by Kronecker-style models: if there exists a~$0 < \tau < 1$ such
that this quantity diverges, then the generated graph contains a dense
subgraph with high probability. In particular, such graphs have unbounded degeneracy.
We used the family of lower bounds derived in 
Corollary~\ref{cor:simple-deg-lower-bound} in order to map out a parametric 
region that is guaranteed to generate asymptotically dense graphs, \cf
Figure~\ref{fig:density}. As observed above, for~$\alpha \geq 1/2$ the graphs
are necessarily dense and the plot nicely exhibits this trend of very small
dense subgraphs as~$\alpha$ tends towards~$1/2$. Note that in the whole range,
the lower bounds for dense subgraphs predict a density that grows
like~$k^{\Theta(1)} c^k$ with~$c$ typically around~$1.2$.
For such moderately exponential functions it is
unsurprising that experimental approaches have failed to identify dense subgraphs:
for typical ranges of~$k$, the polynomial terms easily dominate.
Of the undirected graphs listed
in Table~\ref{tab:rmat-converted}, we find that the fitted parameters for
\netw{AS-Routeviews}, \netw{Bio-Proteins}, and \netw{AS-Newman} fall into
a regime which generates asymptotically dense subgraphs.

%
\subsection{Higher Hamming slices are degenerate}\label{sec:degen-sparse}

\noindent
We saw in the previous section that Kronecker-style models often generate asymptotically
dense graphs. Our proof located this density in the lower Hamming slices and the
question whether the \emph{higher} slices are sparse arises naturally. Here we show
not only that the higher slices are often sparse, we show that they exhibit a sparse
\emph{structure}: if we iteratively remove vertices of low degree, this process will
remove all vertices in~$\mathcal F_{\geq \tau k}$, for some fixed~$\tau$ depending
on the input parameters. We can rephrase this idea in terms of the core-structure of the
generated graph as follows:

\def\stau{{\bar \tau}}
\begin{lemma}\label{lemma:degeneracy}
  Fix parameters~$\twobytwo(\alpha,\beta,\beta,\delta), \mu$. Let~$\stau$ be such that for every~$\tau \geq \stau$,
  we have~$\E[\deg_{\leq \tau k}(x) \mid x \in \mathcal F_{\tau k}] \leq 2^{-\Theta(k)}$. Then there exists~$c \in \R$
  such that the $c$-core of the resulting graphs lies in~$\mathcal F_{< \stau k}$ with high probability.
\end{lemma}
\begin{proof}
  Assume~$\tau \geq \stau$ and let~$x \in \mathcal F_{\tau k}$. We apply the
  multiplicative Chernoff--Hoeffding bound and obtain that
  \[
    \P[\deg_{\leq \tau k}(x) \geq c\mu] \leq \Big( \frac{e}{c\mu} \E[\deg_{\leq \tau k}(x)] \Big)^{c\mu} \!\!.
  \]
  Accordingly, the expected number of vertices in~$\mathcal F_{\tau k}$ with more than~$c\mu$ neighbors
  in~$\mathcal F_{\leq \tau k}$ is at most
  \[
    |\mathcal F_{\tau k}| \Big( \frac{e}{c\mu} \E[\deg_{\leq \tau k}(x)] \Big)^{c\mu}
    \! \leq 2^k \Big( \frac{e}{c\mu} \E[\deg_{\leq \tau k}(x)] \Big)^{c\mu} \!\!.
  \]
  Since~$\E[\deg_{\leq \tau k}(x)] = 2^{-\Theta(k)}$, we can choose~$c$ high enough such that
  this expected value is upper-bounded by~$2^{-k}$ and the claim follows from Markov's inequality.
\end{proof}

\noindent
Using the upper bounds established in Corollary~\ref{cor:deg-upper-bound} we can map out what 
fraction of the higher Hamming slices can be expected to form only sparse connections with higher
slices. Figure~\ref{fig:degeneracy} demonstrates how generator matrices in which~$\beta$ is (roughly) at
least as large as~$\alpha/2$ will result in graphs in which a significant fraction of the vertices
fall outside the denser cores. Concerning the networks listed in Table~\ref{tab:rmat-converted}, we find
that the initiator matrices of \netw{AS-Routeviews} and \netw{AS-Newman} will generate graphs in which the vertices
of~$\mathcal F_{\geq .49k}$ have small degree into higher Hamming slices;
for \netw{Bio-Proteins} the slices~$\mathcal F_{\geq .61k}$, and for \netw{Email-Inside} 
the slices~$\mathcal F_{\geq .65k}$.

\begin{figure}[t]
  \centering
  \hspace*{-6pt}
  \includegraphics[scale=.6]{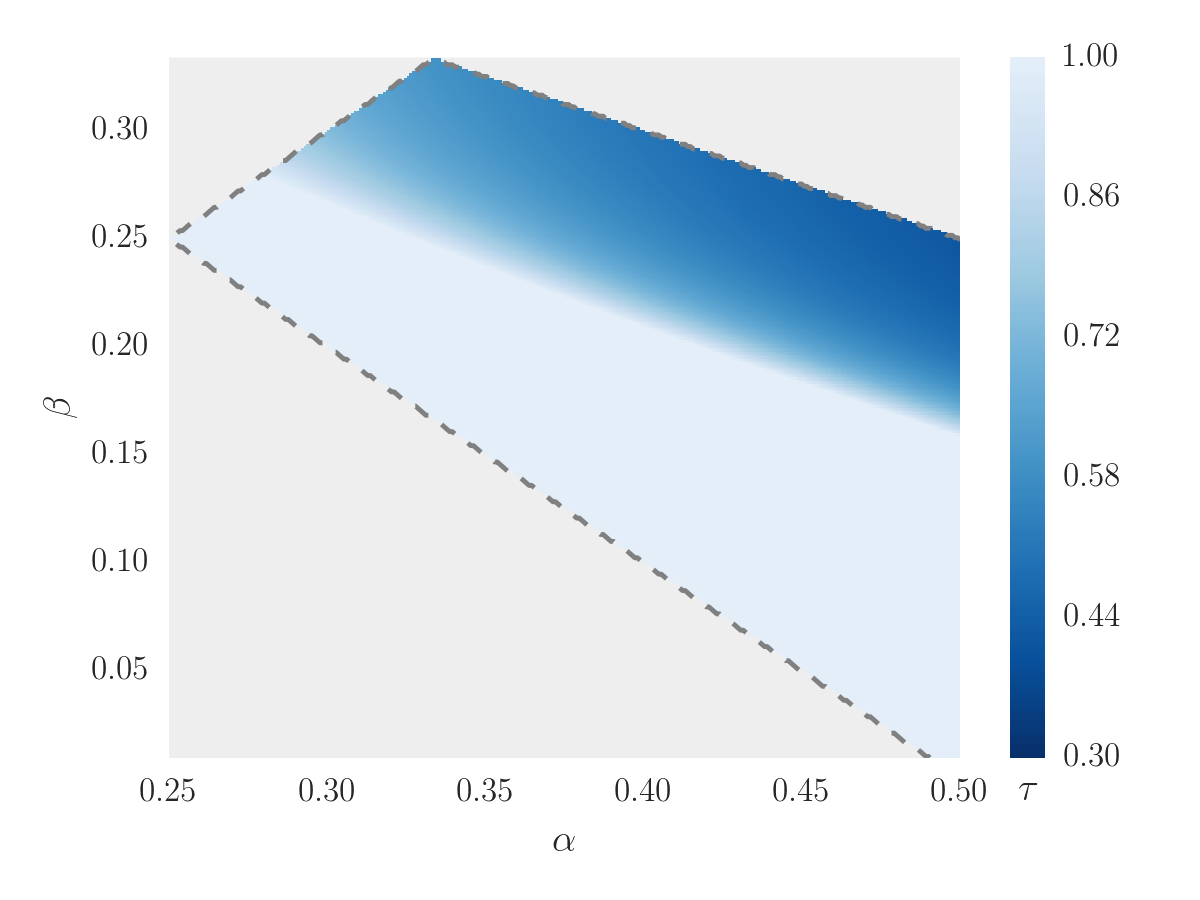}
  \caption{%
    Fraction of Hamming slices~$\mathcal F_{\geq \tau k}$ that fall outside any $c$-core
    in the sense of Lemma~\ref{lemma:degeneracy} in dependence on the input parameters, determined
    via the upper bounds provided by Corollary~\ref{cor:deg-upper-bound}.
    The light area contains those graphs whose core structure does not correlate with the ordering
    by Hamming weight, in particular \Erdos--\Renyi graphs at~$\alpha = \beta = \delta = 0.25$.
    \label{fig:degeneracy}
  }
\end{figure}

\section{Conclusions}
\noindent
In this paper we obtain two asymptotic results pertaining to the structure of
Kronecker-style graph models: (1) a characterization of the conditions under
which variants of this model family are asymptotically equivalent, and (2)
some members of this family can produce deep core structures and dense
subgraphs which are atypical for real-world networks. The latter had not been
detected by empirical methods and our asymptotic bounds provide a putative
explanation for this fact: the scales at which dense subgraphs and deep cores
become apparent lie beyond the usual experimental settings. 

In the other extreme,
our analysis of the arc probabilities in \SKG and \RMAT led to a
statistic that can reveal the difference between these models already at small
graph sizes. This calls into question the common approach of treating these
models as interchangeable and demands further study.

We conclude that the type of analysis presented here has the capacity to further
our understanding of network models beyond currently observable statistics.

\bibliographystyle{IEEEtran}
\bibliography{sections/bib}


\end{document}